\documentclass[letterpaper]{scrartcl}

\usepackage[amsthm]{pamas}

\usepackage{todonotes}
\presetkeys{todonotes}{inline}{}

\usepackage{MnSymbol} 

\newcommand{\romanenumi}{\renewcommand{\theenumi}{({\it\roman{enumi}})}
   \renewcommand{\labelenumi}{\theenumi}}

\usepackage{pgflibraryshapes} 
\usetikzlibrary{positioning} 
\usetikzlibrary{shapes.multipart} 
\usetikzlibrary{arrows,decorations.pathreplacing}

\newcommand{\dom}[1][]{\ifthenelse{\equal{#1}{}}{\overline{D}}{\overline(D)(#1)}}

\usepackage{pifont}
\newlength{\wordlength}

\def\nd{3em} 	
\def\ra{0.75em} 
\tikzset{node distance = \nd, 
		vertex/.style={circle,draw,minimum size=2*\ra, inner sep=1pt}
		}
		
\newcommand{\set}[1]{\{#1\}}

\newtheorem{remark}{Remark}

\title{An Ordinal Minimax Theorem}
\author{
Felix Brandt\\
TU M\"unchen\\Germany
\and
Markus Brill\\
Duke University\\USA
\and
Warut Suksompong\\
Stanford University\\USA
}

\date{}

\begin{document}
	
	\maketitle
	
\begin{abstract}
In the early 1950s Lloyd Shapley proposed an ordinal and set-valued solution concept for zero-sum games called \emph{weak saddle}. We show that all weak saddles of a given zero-sum game are interchangeable and equivalent. As a consequence, every such game possesses a unique set-based value. 
\end{abstract}

 \noindent\textbf{JEL Classifications Code: }C72

\section{Introduction}

One of the earliest solution concepts considered in game theory are \emph{saddle points}, combinations of actions such that no player can gain by deviating~\citep[see, \eg][]{vNM47a}. In two-player zero-sum games, every saddle point happens to coincide with the optimal outcome both players can guarantee in the worst case and thus enjoys a very strong normative foundation. Unfortunately, however, saddle points are not guaranteed to exist. This situation can be rectified by the introduction of \emph{mixed}---\ie randomized---strategies, as first proposed by \citet{Bore21a}. \Citet{vNeu28a} proved that every zero-sum game contains a mixed saddle point, or equilibrium. While equilibria need not be unique, they maintain two appealing properties of saddle points: \emph{interchangeability} (any combination of equilibrium strategies for either player forms an equilibrium) and \emph{equivalence} (all equilibria yield the same expected payoff).

Mixed equilibria have been criticized for resting on demanding epistemic assumptions such as the expected utility axioms by \citet{vNM47a}. See, for example, \citet[][pp.~74--76]{LuRa57a} and \citet{Fish78b}.
As \citeauthor{Auma87a} puts it:
 ``When randomized strategies are used in a strategic game, payoff must be replaced by expected payoff. Since the game is played only once, the law of large numbers does not apply, so it is not clear why a player would be interested specifically in the mathematical expectation of his payoff''~\citep[][p.~63]{Auma87a}.

 \citet{Shap53a,Shap53b} showed that the existence of saddle points can also be guaranteed by moving to \emph{minimal sets} of actions rather than randomizations over them.\footnote{The main results of the 1953 reports later reappeared in revised form~\citep{Shap64a}.}  \citeauthor{Shap53a} defines a \emph{generalized saddle point (GSP)} to be a tuple of subsets of actions for each player that satisfies a simple external stability condition:
 Every action not contained in a player's subset is dominated by some action in the set, given that the other player chooses actions from his set. A GSP is minimal if it does not contain another GSP. Minimal GSPs, which Shapley calls \emph{saddles}, also satisfy internal stability in the sense that no two actions within a set dominate each other, given that the other player chooses actions from his set.
 While Shapley was the first to conceive GSPs, he was not the only one. Apparently unaware of Shapley's work, \citet{Samu92a} uses the very related concept of a \emph{consistent pair} to point out epistemic inconsistencies in the concept of iterated weak dominance. Also, \emph{weakly admissible sets} as defined by \citet{McOr76a} in the context of spatial voting games and the \emph{minimal covering set} as defined by \citet{Dutt88a} in the context of majority tournaments are GSPs \citep{DuLe96a}.\footnote{GSPs have also been considered in the context of general normal-form games \citep[see, \eg][]{DuLe96b,BBFH09a,BBFH10c,BrBr12b}}

In this paper, we consider GSPs with respect to weak dominance. An action weakly dominates another action if it always yields at least as much utility.
\citet[][p.~10]{Shap64a} notes that no general uniqueness result is available for this type of saddle. Later, uniqueness has been shown for restricted classes of zero-sum games, namely tournament games \citep{Dutt88a} and confrontation games  \citep{DuLe96a}.
We show that all weak saddles of a given zero-sum game are interchangeable and equivalent. This implies the above-mentioned uniqueness results and shows that every zero-sum game possesses a unique set-based value. Our result can be interpreted as an ordinal variant of the minimax theorem.

\section{Preliminaries}

A finite two-player zero-sum game is given by a matrix $A=(a_{i,j})_{i \in R, j \in C}$. The finite set $R$ of rows represents the row player's actions, and the finite set $C$ of columns represents the column player's actions. If the row player chooses action $r\in R$, and the column player chooses action $c\in C$, then the \emph{payoff} (or \emph{utility}) of the row player is given by the entry $a_{r,c}$ of the matrix, while the payoff of the column player is given by $-a_{r,c}$.
For nonempty subsets $R'\subseteq R$ and $C' \subseteq C$, $A|_{R' \times C'}$ denotes the \emph{subgame} in which the row player has action set $R'$ and the column player has action set $C'$. 

An action $r_1 \in R$  \emph{weakly dominates} another action $r_2 \in R$ with respect to a set $C' \subseteq C$ of columns, denoted $r_1 \ge_{C'} r_2$, if $a_{r_1,c} \ge a_{r_2,c}$ for all $c \in C'$.\footnote{What we call weak dominance here is sometimes also called \emph{very weak dominance} \citep[see, \eg][]{LeSh08a}. \label{fn:terminology}} Similarly, an action $c_1 \in C$  \emph{weakly dominates} another action $c_2 \in C$ with respect to a set $R' \subseteq R$ of rows, denoted $c_1 \le_{R'} c_2$, if $-a_{r,c_1} \ge -a_{r,c_2}$ (and thus $a_{r,c_1} \le a_{r,c_2}$) for all $r \in R'$. \emph{Strict dominance} is defined analogously, with the weak inequalities replaced by strict inequalities. 

Dominance relations can be extended to sets of actions as follows. 
A set $R_1$ of rows \textit{weakly} (resp. \emph{strictly}) \emph{dominates} a set $R_2$ of rows with respect to $C'\subseteq C$ if for every row $r_2\in R_2$, there exists a row $r_1\in R_1$ such that $r_1$ weakly (resp. strictly) dominates $r_2$ with respect to $C'$. We denote this by $R_1 \geq_{C'} R_2$ (resp. $R_1>_{C'} R_2$). Dominance between sets of columns is defined analogously, and denoted $C_1 \leq_{R'} C_2$ (for weak dominance) and $C_1 <_{R'} C_2$ (for strict dominance).

We are now prepared to define \emph{saddles}, which are based on the notion of a \textit{generalized saddle point (GSP)} \citep{Shap53a,Shap53b,Shap64a}. 
Given a subset $R' \subseteq R$ of rows and a subset $C' \subseteq C$ of columns, the product $R' \times C'$ is a \textit{weak GSP} if $R' \ge_{C'} R\backslash R'$ and $C' \le_{R'} C\backslash C'$. Furthermore, the product $R'\times C'$ is a \textit{weak saddle} if it is a weak GSP and no proper subset of it is a weak GSP.\footnote{Weak saddles have been called \emph{very weak saddles} by \citet{BBFH10c}; see also Footnote \ref{fn:terminology}. In some papers \citep[\eg][]{DuLe96a,DuLe01a,BBFH09a,BBFH10c}, the dominance used for weak saddles requires at least one strict inequality. In the context of confrontation games (see \coref{cor:confrontationunique}), where weak saddles have usually been considered, both notions of weak saddles coincide. \cite{Shap53a,Shap53b,Shap64a} defines weak saddles as we do here. It is easily seen that our theorem does not hold for weak saddles that require at least one strict inequality (see, for example, the restriction to the first two rows and columns of game $A_2$ in \figref{fig:examples}).} Strict GSPs and strict saddles are defined analogously. 

\begin{figure}[htb]
\begin{align*}
A_1 = \left( \begin{array}{ccccc}
2 & 1 & 0 & 1 & 2 \\
0 & 3 & 4 & 4 & 1 \\
0 & 2 & 2 & 1 & 2 \\
2 & 1 & 0 & 2 & 1 
\end{array} \right)
\hspace{1cm}
A_2 = \left( \begin{array}{rrr}
0 & 0 & 0 \\
0 & 1 & -1 \\
0 & -1 & 1
\end{array} \right)
\hspace{1cm}
A_3 = \left( \begin{array}{ccccc}
2 & 2 & 1 & 3 & 2 \\
2 & 4 & 0 & 0 & 2 \\
1 & 3 & 3 & 4 & 1 \\
2 & 3 & 1 & 3 & 2 \\
1 & 0 & 2 & 2 & 0
\end{array} \right)
\end{align*}
\caption{
Three example zero-sum games. For each game, the rows and columns are labeled $r_1,r_2,\ldots$ and $c_1,c_2,\ldots$, respectively. The game $A_1$ contains one weak saddle: $\set{r_1,r_2} \times \set{c_1,c_2,c_3}$.
The game $A_2$ contains a saddle point $\set{r_1} \times \set{c_1}$. This saddle point is the unique pure Nash equilibrium and the unique weak saddle of this game. Moreover, $(\frac{1}{2}r_2 + \frac{1}{2}r_3,\frac{1}{2}c_2 + \frac{1}{2}c_3)$ is a (mixed) Nash equilibrium of~$A_2$. 
The game $A_3$ contains four weak saddles: 
$\set{r_1, r_3} \times \set{c_1,c_3}$, 
$\set{r_1, r_3} \times \set{c_3,c_5}$,
$\set{r_3, r_4} \times \set{c_1,c_3}$, and
$\set{r_3, r_4} \times \set{c_3,c_5}$.
For all three games, the product of all rows and all columns is the unique strict saddle.
}
\label{fig:examples}
\end{figure}

In contrast to strict saddles, weak saddles are extensions of saddle points in the sense that every saddle point constitutes a weak saddle. 
Since the product $R \times C$ containing all actions is a trivial weak and strict GSP of any game, weak and strict saddles are guaranteed to exist. While strict saddles have been shown to be unique in zero-sum games (see \coref{cor:strictunique}), this is not the case for weak saddles. 
It is noteworthy that saddles generally \emph{cannot} be found by the iterated elimination of (weakly or strictly) dominated actions.\footnote{While the subgames generated by iteratively eliminating dominated strategies are GSPs, these GSPs need not be minimal.}
See \figref{fig:examples} for examples.

\section{The Result}

In this section, we prove that weak saddles in zero-sum games are interchangeable and equivalent. We begin with a lemma.

\begin{lemma}
\label{lemma:extstableexpand}
Consider a zero-sum game $A$ with row set $R$ and column set $C$. Let $R_1\subseteq R_2\subseteq R$ and $C_1\subseteq C_2\subseteq C$. Suppose that $R_2\times C_2$ is a weak GSP. Then, $R_1\times C_1$ is a weak GSP if and only if $R_1\times C_1$ is a weak GSP in $A|_{R_2\times C_2}$.
\end{lemma}

\begin{proof}
The ``only if'' part follows straightforwardly from the definitions. For the ``if'' part, suppose that $R_1\times C_1$ is a weak GSP in $A|_{R_2\times C_2}$. We will show that $R_1\geq_{C_1} R\backslash R_1$; the argument for column domination is analogous. Consider an arbitrary row $r\in R\backslash R_1$. If $r\in R_2\backslash R_1$, then since $R_1\times C_1$ is a weak GSP in $A|_{R_2\times C_2}$, there exists a row $r'\in R_1$ such that $r'\geq_{C_1} r$. Otherwise, we have $r\in R\backslash R_2$. Since $R_2\times C_2$ is a weak GSP in $A$, there exists a row $r'\in R_2$ such that $r'\geq_{C_2} r$, and in particular $r'\geq_{C_1} r$. But since $R_1\times C_1$ is a weak GSP in $A|_{R_2\times C_2}$, there exists a row $r''\in R_1$ such that $r''\geq_{C_1} r'$. It follows that $r''\geq_{C_1} r$, and hence $R_1\geq_{C_1} R\backslash R_1$.
\end{proof}

We are now ready to prove the main theorem.

\begin{theorem}
\label{thm:ordinalminimaxthm}
Let $A$ be a zero-sum game with weak saddles $R_1\times C_1$ and $R_2\times C_2$. Then the following statements are true.
\begin{enumerate}\romanenumi
\item $R_1\times C_2$ and $R_2\times C_1$ are also weak saddles (interchangeability).
\item The subgame $A|_{R_2 \times C_2}$ can be derived from $A|_{R_1 \times C_1}$ by permuting the rows and columns (equivalence). In particular, $|R_1|=|R_2|$ and $|C_1|=|C_2|$, and the multisets of entries of $A|_{R_1\times C_1}$ and $A|_{R_2\times C_2}$ are the same.
\end{enumerate}
\end{theorem}

\begin{proof}
The general proof structure is as follows. We first show both statements for the case in which every row of one saddle dominates some row of the other saddle with respect to the columns of the first saddle (and similarly for the columns). We then consider the more difficult case in which a saddle contains an ``idle'' action that does not dominate any action of the other saddle. It turns out that this violates minimality of the saddle and is therefore not possible.
	
Suppose that $R_1\times C_1$ and $R_2\times C_2$ are weak saddles, with $|R_1|=p_1, |R_2|=p_2, |C_1|=q_1$, and $|C_2|=q_2$. Since $R_1\times C_1$ is a weak GSP, for each row $r_2\in R_2$ there exists a row $r_1\in R_1$ such that $r_1\geq_{C_1} r_2$. (If $r_2\in R_1\cap R_2$, then $r_2\geq_{C_1} r_2$.) Hence there exists a function $f_1:[p_2]\rightarrow[p_1]$ such that 
$f_1(i)\in R_1$ and $f_1(i)\geq_{C_1} i$ for every row $i\in R_2$, where $[n]$ denotes the set $\{1,\ldots,n\}$. Similarly, there exist functions $f_2:[p_1]\rightarrow[p_2]$, $g_1:[q_2]\rightarrow [q_1]$, and $g_2:[q_1]\rightarrow[q_2]$ such that 
$f_2(i)\in R_2$ and $f_2(i)\geq_{C_2} i$ for every row $i\in R_1$, 
$g_1(j)\in C_1$ and $g_1(j)\leq_{R_1} j$ for every column $j\in C_2$, and 
$g_2(j)\in C_2$ and $g_2(j)\leq_{R_2} j$ for every column $j\in C_1$. 

Suppose that $f_1,f_2,g_1$, and $g_2$ are all bijections (which in particular implies that $p_1=p_2$ and $q_1=q_2$.) Then the rows in $R_2\times C_1$ are dominated by the rows in $R_1\times C_1$, one by one. Hence $\sum_{i\in R_1, j\in C_1} a_{ij} \geq \sum_{i\in R_2, j\in C_1} a_{ij}$. By putting together these inequalities for $R_1\times C_1$, $R_2\times C_1$, $R_2\times C_2$, and $R_1\times C_2$, we get 
\[\sum_{i\in R_1, j\in C_1} a_{ij} \geq \sum_{i\in R_2, j\in C_1} a_{ij} \geq \sum_{i\in R_2, j\in C_2} a_{ij} \geq \sum_{i\in R_1, j\in C_2} a_{ij} \geq \sum_{i\in R_1, j\in C_1} a_{ij}\text{.}\] 
It follows that equality holds everywhere and that rows are only dominated by identical rows and that columns are only dominated by identical columns. As a consequence, $R_2\times C_2$ can be derived from $R_1\times C_1$ by permuting the rows and columns. Moreover, $R_1\times C_2$ and $R_2\times C_1$ can also be derived from $R_1\times C_1$ by permuting the rows and columns, and one can check that they are saddles in $R\times C$ as well. Hence both (i) and (ii) hold in this case.

Suppose now that at least one of $f_1,f_2,g_1$, and $g_2$ is not a bijection. Then at least one of them is not a surjection. Indeed, if for instance $p_1<p_2$, then $f_2$ is not a surjection. If $p_1=p_2$ and $q_1=q_2$, and any of the functions $f_1,f_2,g_1,g_2$ is not a bijection, then it is also not a surjection. Assume without loss of generality that $f_1$ is not a surjection, \ie $R_1\times C_1$ contains an idle row that does not dominate any row in $R_2$ with respect to $C_1$. 
We will show that $R_1\times C_1$ is not inclusion-minimal, \ie there exists a proper subset $R_1'\times C_1'\subset R_1\times C_1$ such that $R_1'\times C_1'$ is a weak GSP. But since $R_1\times C_1$ is a weak GSP, by Lemma~\ref{lemma:extstableexpand} we only need to show that there exists a proper subset $R_1'\times C_1'\subset R_1\times C_1$ such that $R_1'\times C_1'$ is a weak GSP in $A|_{R_1\times C_1}$.

\begin{figure}[htb]
	\centering
	\scalebox{0.4}{
	\begin{tikzpicture}
		[
		cell/.style={rectangle,draw=black,fill=black!20,thick,inner sep=1.2em},
		greater/.style={->, shorten >=1pt, >=triangle 60,semithick}
		]

	  \draw (0,0) -- (12,0) -- (12,-12) -- (0,-12) -- (0,0);
	  \draw (0,-6.5) -- (7,-6.5) -- (7,0);
	  \draw (5,-12) -- (5,-5) -- (12,-5);
	  
  	\draw [decorate,decoration={brace,amplitude=20pt,mirror},xshift=-50pt,yshift=0pt]
  	(-1,0) -- (-1,-6.5) node [black,midway,xshift=-1.5cm] 
  	{\Huge $R_1$};
	
  	\draw [decorate,decoration={brace,amplitude=20pt},xshift=50pt,yshift=0pt]
  	(13,-5) -- (13,-12) node [black,midway,xshift=1.5cm] 
  	{\Huge $R_2$};

  	\draw [decorate,decoration={brace,amplitude=20pt},xshift=0pt,yshift=30pt]
  	(0,1) -- (7,1) node [black,midway,yshift=1.5cm] 
  	{\Huge $C_1$};
	
  	\draw [decorate,decoration={brace,amplitude=20pt,mirror},xshift=0pt,yshift=-30pt]
  	(5,-13) -- (12,-13) node [black,midway,yshift=-1.5cm] 
  	{\Huge $C_2$};
	
	\draw (-1,-1.5) node{\Huge $i'$};
	\draw (-1.5,-3) node{\Huge $f_1(i)$}
	++(0.5,-5) node{\Huge $i$};
	
	 \draw (2,1) node{\Huge $j$}
	 ++(2.5,0) node{\Huge $g_1(j')$}
	 ++(6,0) node{\Huge $j'$};
	 
	 \draw (2,-13) node{\Huge $j$}
	 ++(6.5,0) node{\Huge $g_2(j)$};
	 \draw (10.5,-13) node{\Huge $j'$};
	 
	 \draw (13,-1.5) node{\Huge $i'$}
	 ++(0.5,-9) node{\Huge $f_2(i')$};
	 \draw (13,-8) node{\Huge $i$};

	 \draw (10.5,-13) node{\Huge $j'$};

	 \draw (2,-3) node[cell](a){}
	 ++(0,-5) node[cell](b){}
	 ++(6.5,0) node[cell](c){};

	 \draw (10.5,-10.5) node[cell](d){}
	 ++(0,9) node[cell](e){}
	 ++(-6,0) node[cell](f){};

	 \draw [greater] (a) to (b); 
	 \draw [greater] (b) to (c); 
	 \draw [greater] (d) to (e); 
	 \draw [greater] (e) to (f); 
	 
	\end{tikzpicture}
	} 
	\caption{Construction in the proof of \thmref{thm:ordinalminimaxthm}. An arrow from one entry to another indicates that the former is greater than or equal to the latter.}
	\label{fig:cycle}
\end{figure}
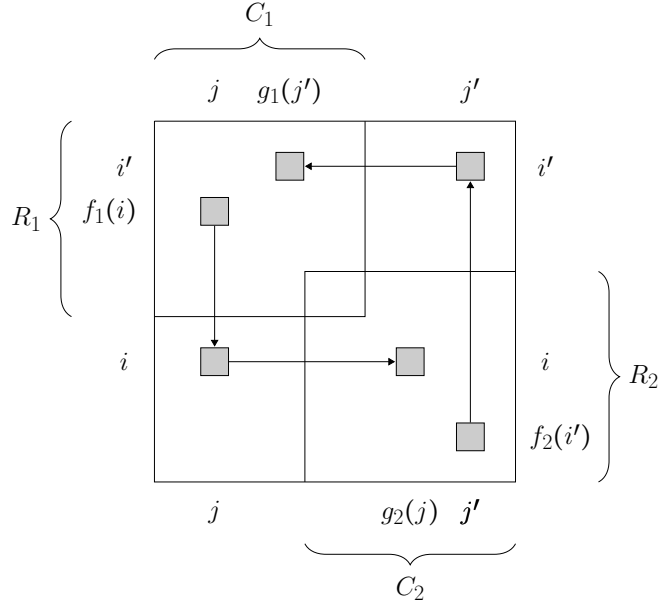

We index the entries of $A|_{R_1\times C_1}$ by $(x_{i,j})_{i \in [p1], j \in [q1]}$, and the entries of $A|_{R_2\times C_2}$ by $(y_{i,j})_{i \in [p2], j \in [q2]}$. We have $x_{f_1(i),j}\geq y_{i,g_2(j)}$ for all $(i,j)\in [p_2]\times [q_1]$ and $x_{i,g_1(j)}\leq y_{f_2(i),j}$ for all $(i,j)\in [p_1]\times [q_2]$. Hence for all $j\in [q_2]$ and all $k\in[q_1]$ such that $g_2(k)=j$, we have $x_{f_1(i),k}\geq y_{i,j}$ for all $i\in [p_2]$. 
Define $x_{f_1(i),g_2^{-1}(j)}:=\{x_{f_1(i),k}\mid g_2(k)=j\}$, and for two sets $S,T$, write $S\geq T$ if and only if $s\geq t$ for all $s\in S$ and $t\in T$. Using this notation, $x_{f_1(i),g_2^{-1}(j)}\geq y_{i,j}$ for all $(i,j)\in [p_2]\times [q_2]$. 
Similarly, we have that $x_{f_2^{-1}(i),g_1(j)}\leq y_{i,j}$ for all $(i,j)\in[p_2]\times [q_2]$. See \figref{fig:cycle} for an illustration.
It follows that 
\begin{equation}
\tag{$\ast$}
x_{f_1(i),g_2^{-1}(j)}\geq x_{f_2^{-1}(i),g_1(j)}
\end{equation}
for all $(i,j)\in[p_2]\times[q_2]$. This inequality will be leveraged later in the proof. 
If $f_2^{-1}(i)=\emptyset$ for some $i$, the inequality is meaningless for that index $i$ and we may remove the index from consideration. A similar statement holds for~$g_2$. We may relabel the remaining indices as $1,\ldots,p_3$ and $1,\ldots,q_3$, respectively, so that $[p_3]=\text{Im}(f_2)$ and $[q_3]=\text{Im}(g_2)$, where $\text{Im}(f)$ is the image of a function $f$. We have $\bigcup_{i\in[p_3]}f_2^{-1}(i)=[p_1]$ and $\bigcup_{j\in [q_3]}g_2^{-1}(j)=[q_1]$.

Consider the product $S=\text{Im}(f_1\circ f_2)\times\text{Im}(g_1\circ g_2)$. Since $f_1$ is not surjective, we have that $S$ is a proper subset of $R_1\times C_1$. Hence it suffices to show that there exists a subset of~$S$ that is a weak GSP in $A|_{R_1\times C_1}$.

We now define a directed graph $\mathcal{G}_R$ as follows.
The nodes of $\mathcal{G}_R$ are given by $1,2,\ldots,p_3$. For each $i$, we include a directed edge from $f_2(f_1(i))$ to $i$. Let $S_R$ be the set of nodes in $\mathcal{G}_R$ that belong to a directed cycle. (A self-loop counts as a directed cycle.) Each node in $\mathcal{G}_R$ has exactly one incoming edge, and any node in $\mathcal{G}_R$ can be reached from a node in $S_R$. Similarly, we define a directed graph $\mathcal{G}_C$ with nodes $1,2,\ldots,q_3$. For each $i$, we include a directed edge from $i$ to $g_2(g_1(i))$. Each node in $\mathcal{G}_C$ has exactly one outgoing edge, and any node in $\mathcal{G}_C$ can reach a node in $S_C$, where $S_C$ is the set of nodes in $\mathcal{G}_C$ that belong to a directed cycle.

Suppose that $\mathcal{G}_R$ contains an edge $i_2\rightarrow i_1$ and $\mathcal{G}_C$ contains an edge $j_1\rightarrow j_2$. Then $f_1(i_1)\in f_2^{-1}(i_2)$ and $g_1(j_1)\in g_2^{-1}(j_2)$.  From ($\ast$), we have $x_{f_2^{-1}(i_1),g_1(j_2)}\leq x_{f_1(i_1),g_2^{-1}(j_2)}$ and $x_{f_2^{-1}(i_2),g_1(j_1)}\leq x_{f_1(i_2),g_2^{-1}(j_1)}$. Since $x_{f_1(i_1),g_1(j_1)}$ belongs to both $x_{f_1(i_1),g_2^{-1}(j_2)}$ and $x_{f_2^{-1}(i_2),g_1(j_1)}$, we have $x_{f_2^{-1}(i_1),g_1(j_2)}\leq x_{f_1(i_1),g_1(j_1)}\leq x_{f_1(i_2),g_2^{-1}(j_1)}$. In particular, we have \[x_{f_2^{-1}(i_1),g_1(j_2)}\leq x_{f_1(i_2),g_2^{-1}(j_1)}\text.\]

Suppose now that $\mathcal{G}_R$ contains edges $i_n\rightarrow i_{n-1}\rightarrow \cdots\rightarrow i_1$ and $\mathcal{G}_C$ contains edges $j_1\rightarrow j_2\rightarrow\cdots\rightarrow j_n$. Then $f_1(i_k)\in f_2^{-1}(i_{k+1})$ and $g_1(j_k)\in g_2^{-1}(j_{k+1})$ for all $k\in [n-1]$. Applying the same argument as in the $n=2$ case repeatedly, we have $x_{f_2^{-1}(i_1),g_1(j_n)}\leq x_{f_1(i_n),g_2^{-1}(j_1)}$.

We claim that $S':=\bigcup_{i\in S_R}f_1(i)\times\bigcup_{j\in S_C}g_1(j)\subseteq S$ is a weak GSP in $A|_{R_1\times C_1}$. To prove this claim, it suffices to consider row domination; column domination follows analogously. 

For each node $x\in S_C$, define $c(x)=g_2(g_1(x))$, \ie $c$ maps a node to its successor in graph $\mathcal{G}_C$. 
We must show that for any $i\in[p_1]$, there exists $j\in S_R$ such that $x_{f_1(j), g_1(k)}\geq x_{i,g_1(k)}$ for all $k\in S_C$. Since $g_1(k)\in g_2^{-1}(c(k))$, it suffices to show that for any $i\in[p_1]$, there exists $j\in S_R$ such that $x_{f_1(j), g_2^{-1}(c(k))}\geq x_{i,g_1(k)}$ for all $k\in S_C$. Moreover, since $\bigcup_{i\in[p_3]}f_2^{-1}(i)=[p_1]$, we only need to show that for any $i\in[p_3]$, there exists $j\in S_R$ such that $x_{f_1(j), g_2^{-1}(c(k))}\geq x_{f_2^{-1}(i),g_1(k)}$ for all $k\in S_C$.

Let $M$ denote the least common multiple of all the cycle lengths in $S_C$. For any positive integer $n$ and any node $k\in S_C$, there exists a path of length $nM-1$ in $S_C$ (and hence in~$\mathcal{G}_C$) that begins with $c(k)$ and ends with $k$. Since every node in $\mathcal{G}_R$ has one incoming edge, for large enough $n'$ there exists a path of length $n'$ in $\mathcal{G}_R$ that begins with some node $j\in S_R$ and ends with $i$. Taking $n''=nM-1$ for large enough $n$, there exists a path of length $n''$ in $\mathcal{G}_C$ that begins with $c(k)$ and ends with $k$, and a path of length $n''$ in $\mathcal{G}_R$ that begins with $j\in S_R$ and ends with $i$. It follows that $x_{f_1(j), g_2^{-1}(c(k))}\geq x_{f_2^{-1}(i),g_1(k)}$. 
Thus $S'$ is a weak GSP, contradicting the minimality of $R_1\times C_1$.
Hence the case in which at least one of $f_1,f_2,g_1$ and $g_2$ is not a bijection cannot occur, and the proof is complete. 
\end{proof}

Every weak saddle $R'\times C'$ of a given zero-sum game $A$ defines a subgame $A|_{R'\times C'}$. It follows from \thmref{thm:ordinalminimaxthm} that all such subgames are identical up to the permutation of rows and columns. As a consequence, $A|_{R'\times C'}$ could be considered the ``essence'' (or set-based value) of $A$.
Moreover, every Nash equilibrium of $A|_{R' \times C'}$ is a Nash equilibrium of $A$. Therefore, every weak saddle contains the support of a Nash equilibrium and the (von Neumann) value of $A|_{R'\times C'}$ is the same as that of the original game $A$.\footnote{However, game $A_2$ in \figref{fig:examples} shows that there can also be Nash equilibria whose support is disjoint from all weak saddles.}

\section{Consequences and Remarks}

\citet{Shap53a} has shown that every zero-sum game contains a unique \emph{strict} saddle. Shapley's proof crucially relies on the minimax theorem (and the interchangeability of minimax strategies).\footnote{Shapley notes that Hyman Bass proved this statement without making reference to the minimax theorem, but Bass's report is unavailable.} Shapley's result can be obtained as a corollary of \thmref{thm:ordinalminimaxthm} by leveraging the interchangeability of weak saddles.

\begin{corollary}[\citealp{Shap53a}] \label{cor:strictunique}
Every zero-sum game contains a unique strict saddle.
\end{corollary}
\begin{proof}
Let $R_1\times C_1$ and $R_2\times C_2$ be two distinct strict saddles.
We first show that $R_1\cap R_2\ne \emptyset$ and $C_1\cap C_2\ne \emptyset$. Without loss of generality we may assume for contradiction that $R_1\cap R_2= \emptyset$. Every strict GSP is also a weak GSP and therefore contains a weak saddle.
Let $R_1'\times C_1'$ be a weak saddle such that $R_1'\subseteq R_1$ and $C_1' \subseteq C_1$ and let $R_2'\times C_2'$ be a weak saddle such that $R_2'\subseteq R_2$ and $C_2' \subseteq C_2$. \thmref{thm:ordinalminimaxthm} implies that $R_2'\times C_1'$ is also a weak saddle. Note that $R_2' \cap R_1=\emptyset$. Now let $r_1$ be an arbitrary row in $R_1$. Since $R_2'\times C_1'$ is a weak saddle, there must be $r_2\in R_2'$ such that $r_1 \le_{C_1'} r_2$. Since $R_1\times C_1$ is a strict saddle, there has to be $r_3\in R_1$ such that $r_3 >_{C_1} r_2$ (and hence also $r_3 >_{C_1'} r_2$). The strict inequality implies that $r_3\ne r_1$. Alternating these two arguments, we obtain a sequence of rows $r_1$, $r_2$, $\dots$ such that
\[
r_1 \;\le_{C_1'}\; r_2 \;<_{C_1'}\; r_3 \;\le_{C_1'}\; r_4 \; <_{C_1'} \;\dots
\]
All $r_i$ have to be distinct and finiteness of the game implies that we will eventually find a row $r_k$, which is witness to the fact that either $R_2'\times C_1'$ is not a weak saddle or $R_1\times C_1$ is not a strict saddle, both of which are contradictions.

We thus have that $R_1\cap R_2\ne \emptyset$ and $C_1\cap C_2\ne \emptyset$.
An argument using a chain of strict inequalities (similar to the one above) easily shows that strict GSPs are closed under non-empty intersection, \ie $(R_1\cap R_2)\times (C_1\cap C_2)$ is a strict GSP. Hence, at least one of the two original strict GSPs was not minimal, a contradiction.
\end{proof}

A zero-sum game $A$ with row set $R$ and column set $C$ is \emph{symmetric} if $R=C$ (with a slight abuse of notation) and the payoff matrix $A$ is skew-symmetric. In particular, all entries on the main diagonal of $A$ are zero. A \textit{confrontation game} is a symmetric zero-sum game in which zeroes appear \emph{only} on the main diagonal. 
Strengthening a result by \citet{Dutt88a}, \citet{DuLe96a} have shown that every confrontation game contains a unique weak saddle.\footnote{By contrast, Nash equilibria are not unique in confrontation games \citep{LeBr05a}.} This now also follows as a consequence of Theorem \ref{thm:ordinalminimaxthm}. 

\begin{corollary}[\citealp{DuLe96a}]
\label{cor:confrontationunique}
Every confrontation game contains a unique weak saddle.
\end{corollary}

\begin{proof}
Let $A$ be a confrontation game. A weak saddle $R'\times C'$ is called \textit{symmetric} if $R'=C'$. Hence, a weak saddle $R'\times C'$ is symmetric if and only if every row and every column of $A|_{R' \times C'}$ contains exactly one zero.
	
Assume for contradiction that $A$ has two distinct weak saddles, $R_1\times C_1$ and $R_2\times C_2$. 
We consider the following two cases.

\textit{Case 1}: At least one of the two weak saddles is not symmetric. Assume without loss of generality that $R_1\times C_1$ is not symmetric. Since $A$ is skew-symmetric, $C_1 \times R_1$ is also a weak saddle. By the first part of \thmref{thm:ordinalminimaxthm}, $R_1\times R_1$ is a symmetric weak saddle. Hence, $R_1\times R_1$ contains exactly one zero in each row and each column, while $R_1\times C_1$ does not. This contradicts the second part of \thmref{thm:ordinalminimaxthm}.

\textit{Case 2}: Both weak saddles are symmetric. Then $R_1\times C_2$ is an asymmetric saddle, and we obtain a contradiction in the same way as in Case 1.
\end{proof}

We conclude the paper with three remarks.

\begin{remark}
	If all payoffs of a game are pairwise distinct, strict dominance and weak dominance coincide. Thus, every such zero-sum game contains a unique weak saddle.
\end{remark}

\begin{remark}
	\citet{DuLe01a} defined refinements of (weak and strict) saddles based on \emph{mixed} dominance.  Game $A_2$ in \figref{fig:examples} shows that \thmref{thm:ordinalminimaxthm} does \emph{not} hold for mixed weak saddles.
\end{remark}

\begin{remark}
	A zero-sum game may contain an exponential number of weak saddles \citep{BrBr12b}. Thus, computing \emph{all} weak saddles of a game is not feasible in polynomial time. The computational complexity of finding \emph{some} weak saddle of a zero-sum game is open.\footnote{For two-player games that are not zero-sum, finding a weak saddle has been shown to be NP-hard \citep{BBFH10c}.}
\end{remark}

\section*{Acknowledgements}
This material is based on work supported by Deutsche Forschungsgemeinschaft under grants {BR~2312/7-2} and {BR~2312/10-1}, by a Feodor Lynen research fellowship of the Alexander von Humboldt Foundation, and by the MIT-Germany program. The authors thank Vincent Conitzer and Paul Harrenstein for helpful discussions.

\end{document}